\documentclass{LMCS}

\def\doi{7 (2:12) 2011}
\lmcsheading%
{\doi}
{1--21}
{}
{}
{Jan.~14, 2010}
{Nov.\phantom~16, 2011}
{}

\usepackage{\jobname}
\usepackage{enumerate,hyperref}

\title{%
Model Checking CTL is Almost Always Inherently Sequential\rsuper*%
}

\author[O.~Beyersdorff]{Olaf~Beyersdorff\rsuper a}
\address{{\lsuper{a,b,e,f}}Theoretical Computer Science, Leibniz
University of Hannover, Germany}
\email{\{beyersdorff, meier, thomas, vollmer\}@thi.uni-hannover.de}%

\author[A.~Meier]{Arne~Meier\rsuper b}
\address{\vskip-6 pt}

\author[M.~Mundhenk]{Martin~Mundhenk\rsuper c}
\address{{\lsuper c}Computer Science, University of Jena, Germany}
\email{martin.mundhenk@uni-jena.de}

\author[T.~Schneider]{Thomas~Schneider\rsuper d}
\address{{\lsuper d}Computer Science, Saarland University, Germany}
\email{schneider@ps.uni-saarland.de}

\author[M.~Thomas]{Michael~Thomas\rsuper e}
\address{\vskip-6 pt}

\author[H.~Vollmer]{Heribert~Vollmer\rsuper f}
\address{\vskip-6 pt}

\titlecomment{{\lsuper*}A preliminary version of this paper appeared
  in the proceedings of the conference TIME'09~\cite{BMMSTV09}.}
\thanks{
Supported in part by grants DFG VO 630/6-1, VO 630/6-2, DAAD-ARC
D/08/08881, and BC-ARC 1323.
}


\begin{document}

\begin{revision}
  This is a revised and corrected version of the article originally
published on May 17, 2011.
\end{revision}

\begin{abstract}
The model checking problem for \CTL is known to be \P-complete
(Clarke, Emerson, and Sistla (1986)\nocite{clemsi86}, see
Schnoe\-be\-len (2002)\nocite{schnoeb02}).  
We consider fragments of \CTL obtained by
restricting the use of temporal modalities or
the use of negations%
---restrictions already studied
for \LTL by Sistla and Clarke (1985)\nocite{sicl85} and Markey
(2004)\nocite{mar04}.  
For all these fragments, except for the
trivial case without any temporal operator, we systematically prove model
checking to be either inherently sequential (\P-complete) or 
very efficiently parallelizable (\LOGCFL-complete).
For most fragments, however, model checking for \CTL is already
\P-complete.
Hence our results indicate that, in cases where the combined complexity is of relevance,
approaching \CTL model checking by parallelism cannot be expected to result in any significant speedup.

We also completely determine the complexity of the model checking problem
for all fragments of the extensions \ECTL, \CTLplus, and \ECTLplus.
\end{abstract}

\maketitle
\vfill

\section{Introduction} \label{sect:intro}

Temporal logic was introduced by Pnueli \cite{pnu77} as a formalism to
specify and verify properties of concurrent programs. Computation Tree
Logic (\CTL), the logic of branching time, goes back to Emerson and
Clarke \cite{emcl82} and contains temporal operators for expressing
that an event occurs at some time in the future ($\F$), always in the
future ($\G$), in the next point of time ($\X$), always in the future
until another event holds ($\U$), or as long as it is not released by the occurrence
of another event ($\R$), as well as path quantifiers
($\E,\A$) for speaking about computation paths. 
The full language obtained by these operators and quantifiers is called
\CTLstar \cite{emha86}.
In \CTL, the interaction between the temporal operators and path quantifiers is
restricted.  
The temporal operators in \CTL
are obtained by path quantifiers followed directly by any temporal operator,
\emph{e.g.}, $\AF$ and $\AU$ are \CTL-operators.
Because they start with the universal path quantifier,
they are called \emph{universal \CTL-operators}.
Accordingly, $\EX$ and $\EG$ are examples for \emph{existential \CTL-operators}.

Since properties are largely verified automatically, the computational
complexity of reasoning tasks is of great interest. 
Model checking (MC)---the problem of verifying whether a given formula holds in a
state of a given model---is one of the most important reasoning tasks
\cite{schnoeb02}. 
It is intractable for \CTLstar (\PSPACE-complete \cite{emle87,schnoeb02}),
but tractable for \CTL 
(complete for polynomial time \cite{clemsi86,schnoeb02}).

Although model checking for \CTL is tractable, its \P-hardness means
that it is presumably not efficiently parallelizable. We therefore
search for fragments of \CTL with a model checking problem of
lower complexity.  
We will consider all subsets of \CTL-operators, 
and examine the complexity of the
model checking problems for all resulting fragments of \CTL.
Further, we consider three additional restrictions
affecting the use of negation and 
study the extensions \ECTL, \CTLplus, and their combination \ECTLplus.

The complexity of model checking for fragments of temporal logics has been
examined in the literature: Markey \cite{mar04} considered satisfiability and model checking for
fragments of Linear Temporal Logic (\LTL). Under systematic
restrictions to the temporal operators, the use of negation, and the interaction of
future and past operators, Markey classified the two decision problems into
\NP-complete, \coNP-complete, and \PSPACE-complete.
Further, \cite{bamuscscscvo07entcs} examined model checking for
all fragments of \LTL obtained by 
restricting
the set of temporal operators and propositional connectives. 
The resulting classification separated cases where model checking 
is tractable from those where it is intractable. For model checking paths in \LTL  
an $\AC{1}(\LOGDCFL)$ algorithm is presented in \cite{KuhFin09}.

Concerning \CTL and its extension \ECTL, our results in this paper 
show that most restricted versions of the model checking problem exhibit
the same hardness as the general problem.
More precisely, we show that apart from the trivial case where
\CTL-operators are completely absent, the complexity of \CTL model
checking is a dichotomy: it is either $\P$-complete or
$\LOGCFL$-complete. 
Unfortunately, the latter case only occurs for a few rather weak
fragments and hence there is not much hope that in practice, model
checking can be sped up by using parallelism---it is inherently sequential.

Put as a simple rule, 
model checking for \CTL is \P-complete
for every fragment that allows to express a universal \emph{and} an
existential \CTL-operator. 
Only for fragments involving the operators $\EX$ and $\EF$ (or
alternatively $\AX$ and $\AG$) model checking is \LOGCFL-complete. 
This is visualized in Figure~\ref{fig:ctl-mc-BF} in Section~\ref{sect:conclusion}.
Recall that \LOGCFL is defined as the class of problems logspace-reducible
to context-free languages,
and $\NL\subseteq \LOGCFL \subseteq \NC2  \subseteq \P$.
Hence, in contrast to inherently sequential \P-hard tasks, 
problems in \LOGCFL have very efficient parallel algorithms.

For the extensions $\CTLplus$ and $\ECTLplus$, the situation is more complex.
In general, model checking $\CTLplus$ and $\ECTLplus$ is
$\DeltaPtwo$-complete \cite{lamasc01}. 
We show that for $T \subseteq \{\A,\E,\X\}$, both model checking problems restricted to operators
from $T$ remain tractable, while for $T \nsubseteq \{\A,\E,\X\}$, they become
$\DeltaPtwo$-complete.
Yet, for negation restricted fragments with only existential or only universal path quantifiers, 
we observe a complexity decrease to $\NP$- resp.\ $\co\NP$-completeness.

This paper is organized as follows: 
Section~\ref{sect:prelim} introduces \CTL, its model checking problems, 
and the non-basics of complexity theory we use. 
Section~\ref{sect:ctl-mc} contains our main results,
separated into upper and lower bounds. 
We also provide a
refined analysis of the reductions between different model checking
problems with restricted use of negation. 
The results are then generalized to extensions of \CTL in Section~\ref{sect:extensions}.
Finally, Section~\ref{sect:conclusion} concludes with a graphical
overview of the results. 

\section{Preliminaries} \label{sect:prelim}

  \subsection{Temporal Logic}
    We inductively define $\CTLstar$-formulae as follows. 
    Let $\Phi$ be a finite set of atomic propositions.
    The symbols used are the atomic propositions in $\Phi$,  
    the constant symbols $\true$, $\false$, 
    the Boolean connectives $\neg$, $\land$, and $\lor$,
    and the temporal operator symbols $\A$, $\E$, $\X$, $\F$,
    $\G$, $\U$, and $\R$.

    $\A$ and $\E$ are called a \emph{path quantifiers}, temporal operators aside
    from $\A$ and $\E$ are \emph{pure temporal operators}. 
    The atomic propositions and the constants $\true$ and $\false$ 
    are \emph{atomic formulae}.  
    There are two kinds of formulae, \emph{state formulae} and \emph{path formulae}.
    Each atomic formula is a state formula, and each state formula is
    a path formula.
    If $\varphi, \psi$ are state formulae and $\chi, \pi$ are path formulae,
    then $\lnot\varphi$, $(\varphi\land\psi)$, $(\varphi\lor\psi)$, 
    $\A\chi$, $\E\chi$ are state formulae,
    and $\lnot\chi$, $(\chi\land\pi)$, $(\chi\lor\pi)$, $\X\chi$,
    $\F\chi$, $\G\chi$, $[\chi\U\pi]$, and $[\chi\R\pi]$
    are path formulae. 
    The set of $\CTLstar$-formulae (or \emph{formulae}) consists of all state formulae. 

    A \emph{Kripke structure} is a triple $K=(W,R,\eta)$, 
        where $W$ is a finite set of states, 
        $R \subseteq W \times W$ a total relation 
              (\emph{i.e.}, for each $w \in W$, there exists a $w'$ such that $(w,w') \in R$), 
        and $\eta \colon W \limplies \powerset{\Phi}$ is a labelling function. 
    A \emph{path} $x$ is an infinite sequence $x=(x_1, x_2, \ldots)
    \in W^\omega$ such that $(x_i,x_{i+1}) \in  R$, for all $i\geq 1$.
    For a path $x=(x_1, x_2, \ldots)$ we denote by $x^i$ the path $(x_i,x_{i+1},\ldots)$.

    Let $K=(W,R,\eta)$ be a Kripke structure, $w \in W$ be a state, and
    $x=(x_1,x_2,\dots) \in W^\omega$ be a path. 
    Further, let $\varphi, \psi$ be state formulae and $\chi, \pi$ be path formulae.
    The truth of a $\CTLstar$-formula w.r.t.\ $K$ is inductively defined as follows:
    \begin{tabbing}
      \hspace{\parindent}\=$M,x$ \= $\models (\varphi\land\psi)$ \= iff \= $ K,w $ \= \kill
      \>$K,w$\>$\models \true           $   \> always,  \\
      \>$K,w$\>$\models \false          $   \> never,  \\
      \>$K,w$\>$\models p               $   \> iff \> $p \in \Phi$ and $p \in \eta(w)$, \\
      \>$K,w$\>$\models \neg\varphi     $   \> iff \> $K,w$ \> $\not\models \varphi$, \\
      \>$K,w$\>$\models (\varphi\land\psi)$ \> iff \> $K,w$ \> $\models \varphi$ and $K,w \models \psi$, \\
      \>$K,w$\>$\models (\varphi\lor\psi)$  \> iff \> $K,w$ \> $\models \varphi$ or $K,w \models \psi$, \\
      \>$K,w$\>$\models \A\chi          $   \> iff \> $K,x$ \> $\models \chi$ for all paths $x=(x_1,x_2,\ldots)$ with $x_1=w$, \\
      \>$K,x$\>$\models \varphi           $ \> iff \> $K,x_1$ \>      $\models \varphi$, \\
      \>$K,x$\>$\models \neg\chi        $   \> iff \> $K,x$ \> $\not\models \chi$, \\
      \>$K,x$\>$\models (\chi   \land\pi )$ \> iff \> $K,x$ \> $\models \chi$ and $K,x \models \pi$, \\
      \>$K,x$\>$\models (\chi   \lor\pi )$  \> iff \> $K,x$ \> $\models \chi$ or $K,x \models \pi$, \\
      \>$K,x$\>$\models \X\chi       $      \> iff \> $K,x_2$ \>  $\models \chi$ \\
      \>$K,x$\>$\models [\chi\U\pi]  $ \> iff \> there is a $k\in \N$ such that $K,x^k \models \pi$ and $K,x^i \,\models \chi$ for $1\leq i<k$.
    \end{tabbing}
    The semantics of the remaining temporal operators is defined via
    the equivalences:
      $\E\chi\equiv\lnot\A\lnot\chi$,
      $\F\chi\equiv [\true\U\chi]$, 
      $\G\chi\equiv\lnot\F\lnot\chi$, and 
      $[\chi\R\pi] \equiv \neg [\neg \chi \U \neg \pi]$.
    A state formula $\varphi$ is 
    \emph{satisfied by a Kripke structure $K$} if there exists $w \in W$ such that
    $K,w \models \varphi$. We will also denoted this by $K\models\varphi$.  

    We use $\CTLstar(T)$ to denote the set of $\CTLstar$-formulae
    using the Boolean connectives $\{\land, \lor, \neg\}$,
    and the temporal operators in $T$ only.
    If $T$ does not contain any quantifiers,
    then including any pure temporal operators in $T$ is meaningless.

    A \emph{$\CTL$-formula} is a $\CTLstar$-formula 
    in which each path quantifier is followed by exactly one pure temporal operator 
    and each pure temporal operator is preceded by exactly one path quantifier.
    The set of $\CTL$-formulae forms a strict subset of the set of all $\CTLstar$-formulae.
    For example, $\AG\EF p $ is a $\CTL$-formula, but $\A(\G\F p \land \F q)$ is not. 
    $\CTL$ is less expressive than $\CTLstar$ \cite{emha85,emha86}.

    Pairs of path quantifiers and pure temporal operators are 
    called \emph{$\CTL$-operators}.  
    The operators $\AX$, $\AF$, $\AG$, $\AU$, and $\AR$
    are \emph{universal} $\CTL$-operators,
    and  $\EX$, $\EF$, $\EG$, $\EU$, and $\ER$ are \emph{existential} $\CTL$-operators. Let $\ALL$ denote the set of all universal and existential \CTL-operators.
    Note that $\A[{\psi}\U{\chi}] \equiv \AF\chi \land \neg \E[{\neg \chi}\U{(\neg \psi \land \neg \chi)}]$, 
    and thus $\E[{\psi}\R{\chi}] \equiv \EG\chi \lor \E[{\chi}\U{(\psi \land \chi)}]$.
    Hence $\{\AX, \AF, \AR\}$ is a minimal set of 
    operators for $\CTL$     (in presence of all Boolean connectives), 
    whereas $\{\AX,$ $\AG,$ $\AU\}$ is not 
    \cite{laroussinie95}.

    By $\CTL(T)$ we denote the set of $\CTL$-formulae using 
    the connectives $\{\land, \lor, \neg\}$ and 
    the \CTL-operators in $T$ only.
    Figure~\ref{fig:CTL-ops-lattice} shows the structure of
    sets of \CTL-operators with respect to their expressive power.
    \begin{figure}
      \centering
      \includegraphics[width=0.5\linewidth]{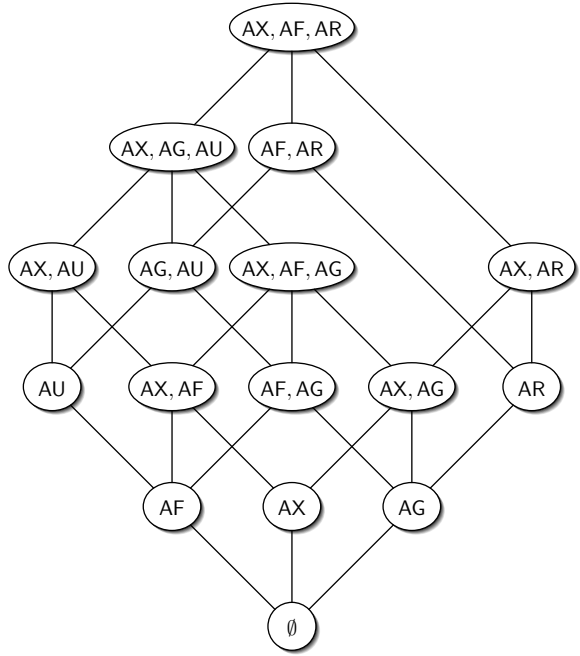}
      \caption{\label{fig:CTL-ops-lattice}%
        The expressive power of $\CTL(T)$.
      }
    \end{figure}

    Moreover, we define the following fragments of $\CTL(T)$.

    \begin{iteMize}{$-$}
      \item
        $\CTLpos(T)$ \quad (positive) \\
        \CTL-operators may not occur in the scope of a negation.
        \par\medskip
     \item
        $\CTLan(T)$ \quad (atomic negation) \\
        Ne\-ga\-tion signs appear only directly in front of atomic propositions.
        \par\medskip
      \item
        $\CTLmon(T)$ \quad (monotone) \\
        No negation signs allowed.
    \end{iteMize}

\noindent This restricted use of negation was introduced and studied
in the context of linear temporal logic, \LTL, by Sistla and Clarke
\cite{sicl85} and Markey \cite{mar04}.  Their original notation was
$\widetilde{\mathrm{L}}(T)$ for $\CTLan(T)$ and $\mathrm{L}^+(T)$ for
$\CTLpos(T)$.
  
  \subsection{Model Checking} 
  Now we define the \emph{model checking problems} 
  for the above mentioned fragments of \CTL.
  Let $\mathcal{L}$ be \CTL, \CTLmon, \CTLan, or \CTLpos.
  \begin{center}  
    \begin{tabularx}{\linewidth}{@{\hspace*{1em}}rX@{\hspace*{1em}}}
      \textit{Problem}:& $\mathcal{L}\text{-}\MC(T)$ \\
      \textit{Input}:  & A Kripke structure $K=(W,R,\eta)$, a state $w\in W$, and an $\mathcal{L}(T)$-formula $\varphi$. \\  
      \textit{Question}: & Does $K,w \models \varphi$ hold?\\
    \end{tabularx}
  \end{center}

  \subsection{Complexity Theory}\label{subsect:complexity}

    We assume familiarity with standard notions of complexity theory 
    as introduced in, \emph{e.g.}, \cite{pap94}.
    Next we will introduce the notions from circuit complexity that we use for our results.
All reductions in this paper are \leqcd-reductions defined as follows: 
A language $A$ is \emph{constant-depth reducible} to $B$, $A\leqcd B$, if there is a logtime-uniform $\AC{0}$-circuit family with oracle gates for $B$ that decides membership in $A$.
That is, there is a circuit family $\mathcal{C}=(C_1,C_2,C_3,\dots)$ such that

\begin{iteMize}{$-$}
\item for every $n$, $C_n$ computes the characteristic function of $A$ for inputs of length $n$,
\item there is a polynomial $p$ and a constant $d$ such that for all input lengths $n$, the size of $C_n$ is bounded by $p(n)$ and the depth of $C_n$ is bounded by $d$,
\item each circuit $C_n$ consists of unbounded fan-in AND and OR
  gates, negation gates, and gates that compute the characteristic
  function of $B$ (the  
  \emph{oracle gates}), 
\item there is a linear-time Turing machine $M$ 
  that can check the structure of the circuit family, \emph{i.e.}, 
  given a tuple $\langle n,g,t,h\rangle$ where $n,g,h$
are binary numbers and
$t\in\{\mathrm{AND},\mathrm{OR},\mathrm{NOT},\mathrm{ORACLE}\}$, $M$
accepts if $C_n$ contains a gate $g$ of type $t$ with predecessor
$h$.
\end{iteMize}

\noindent Circuit families $\mathcal{C}$ with this last property are called
\emph{logtime-uniform} (the name stems from the fact that the time
needed by $M$ is linear in the length of its input tuple, hence
logarithmic in $n$). For background information we refer to \cite{revo97,vol99}. 

We easily obtain the following relations between model checking for
fragments of \CTL with restricted negation:
\begin{lem}\label{prop:simple-cdreds-mc}\label{lem:atomic-negation=no-negation}
  For every set $T$ of \CTL-operators, we have 
  \[
    \CTLmonMC(T)\leqcd\CTLanMC(T) \leqcd\CTLposMC(T).
  \]
  Further, for model checking, atomic negation can be eluded, that is, $\CTLanMC(T) \leqcd \CTLmonMC(T)$.
\end{lem}

  \begin{proof}
    The first part is straightforward, using the identity function as reduction function.
    For the second part,
    let $K=(W,R,\eta)$ be a Kripke structure and let $\varphi$ be 
    a $\CTLan(T)$-formula over the propositions $\Phi=\{p_1,\ldots,p_n\}$.
    Every negation in $\varphi$ appears 
    inside a negative literal.
    We obtain $\varphi'$ by replacing every negative literal
    $\neg p_i$ with a fresh atomic proposition $q_i$. 
    Further define $K'=(W,R,\eta')$, where $\eta'(w)=\eta(w) \cup \{ q_i \mid p_i \notin \eta(w)\}$.
    Obviously, $K,w \models \varphi$ iff $K',w \models \varphi'$ for all $w \in W$.
    The mapping $(K,w, \varphi)\mapsto (K',w,\varphi')$ can be performed by an \AC0-circuit.
  \end{proof}

In Section~\ref{subsect:negation}, we complete the picture by proving
$\CTLposMC(T)\leqcd\!\CTLmonMC(T)$.

The class $\P$ consists of all languages that have a polynomial-time decision algorithm. 
A problem is $\P$-complete if it is in $\P$ and every other problem in $\P$ reduces to it. 
$\P$-complete problems are sometimes referred to as \emph{inherently
sequential}, because $\P$-complete problems most likely
(formally: if $\P\neq\NC{}$) do not possess $\NC{}$-algorithms,
that is, algorithms running in polylogarithmic time on a parallel
computer with a polynomial number of processors.
Formally, $\NC{}$ contains all problems solvable by polynomial-size
polylogarithmic-depth logtime-uniform families of circuits with 
bounded fan-in AND, OR, NOT gates.

There is an $\NC{}$-algorithm for parsing context-free languages, that is, $\CFL\subseteq\NC{}$. 
Therefore, complexity theorists have studied the class $\LOGCFL$ of all problems reducible to context-free languages 
(the name ``\LOGCFL'' refers to the original definition of the class in terms of logspace-reductions, 
however it is known that the class does not change if instead, as everywhere else in this paper, $\leqcd$-reductions are used). 
Hence, $\LOGCFL\subseteq\NC{}$ (even $\LOGCFL\subseteq\NC{2}$, the second level of the \NC{}-hierarchy, where the depth of the occurring circuits is restricted to $O(\log^2n)$).
The class $\LOGCFL$ has a number of different maybe even somewhat surprising characterizations, 
\emph{e.g.}, languages in $\LOGCFL$ are those that can be decided by nondeterministic Turing machines operating in polynomial time 
that have a worktape of logarithmic size and additionally a stack whose size is not bounded.

More important for this paper is the characterization of $\LOGCFL$ as those problems computable by $\SAC{1}$ circuit families, that is, families of circuits that

\begin{iteMize}{$-$}
\item have polynomial size and logarithmic depth,
\item consist of unbounded fan-in OR gates and bounded fan-in AND gates and negation gates, but the latter are only allowed at the input-level,
\item are logtime-uniform (as defined above).
\end{iteMize}

Since the class $\LOGCFL$ is known to be closed under complementation, the
second condition can equivalently be replaced to allow unbounded
fan-in AND gates and restrict the fan-in of OR gates to be bounded. 

To summarize:
\[
  \NC1\subseteq L\subseteq\NL\subseteq\LOGCFL=\SAC{1}\subseteq\NC{2};
\]
and problems in these classes possess very efficient parallel algorithms: they can be solved in time $O(\log^2n)$ on a parallel machine with a tractable number of processors. 
For more background on these and related complexity classes, we refer the reader to \cite{vol99}.

\section{Model Checking CTL and CTL$_\text{pos}$} \label{sect:ctl-mc}

This section contains our main results on the complexity of model
checking for $\CTL$ and $\CTLpos$. 
We defer the analysis of the fragments $\CTLan$ and $\CTLmon$ to
Section~\ref{subsect:negation}, where we will see that their
model-checking problems are computationally 
equivalent to model checking for $\CTLpos$.

While model checking for $\CTL$ in general is known to be polynomial
time solvable and in fact $\P$-complete \cite{clemsi86,schnoeb02}, we
improve the lower bound by showing that only one temporal operator is
sufficient to obtain hardness for $\P$.

\begin{thm}\label{thm1}
  For each nonempty set $T$ of $\CTL$-operators, $\CTLMC(T)$ is
  \P-complete.
  If $T=\emptyset$, then $\CTLMC(T)$ is \NC1-complete.
\end{thm}

If we consider only formulae from $\CTLpos$, 
where no \CTL-operators are allowed inside the scope of a negation, 
the situation changes and the complexity of model checking exhibits a dichotomous behavior.
As long as $\EG$ or $\AF$ are expressible the model checking problem remains \P-complete.
Otherwise, its complexity drops to $\LOGCFL$.

\begin{thm}\label{thm2}
Let $T$ be any set of \CTL-operators. Then $\CTLposMC(T)$ is

\begin{iteMize}{$-$}
\item
  \NC1-complete if $T=\emptyset$,
\item
  \LOGCFL-complete if $\emptyset\subsetneq T \subseteq \{\EX,\EF\}$ 
   or $\emptyset\subsetneq T\subseteq \{\AX,\AG\}$, and 
\item
  \P-complete otherwise.
\end{iteMize}
\end{thm}\medskip

We split the proofs of Theorems~\ref{thm1} and \ref{thm2} into 
the upper and lower bounds in the following two subsections.

  \subsection{Upper Bounds}
  \label{subsect:upper bounds}

    In general, model checking for $\CTL$ is known to be solvable in $\P$ \cite{clemsi86}.
    While this upper bound also applies to $\CTLposMC(T)$ (for every
    $T$), we improve it for positive $\CTL$-formulae  
    using only $\EX$ and $\EF$, or only $\AX$ and $\AG$.
    \begin{prop}\label{prop:LOGCFL_upper_bound}
      Let $T$ be a set of $\CTL$-operators such that 
      $T \subseteq \{\EX,\EF\}$ or $T \subseteq \{\AX,\AG\}$.
      Then $\CTLposMC(T)$ is in $\LOGCFL$.
    \end{prop}
    
    \begin{proof}
      First consider the case $T \subseteq \{\EX,\EF\}$.
      We claim that Algorithm~1 
      recursively decides whether the Kripke structure $K=(W,R,\eta)$ 
      satisfies the $\CTLpos(T)$-formula $\varphi$ in state $w_0 \in W$. 
      There, $S$ is a stack that stores pairs $(\varphi,w) \in \CTLpos(T) \times W$ and $R^\star$ denotes the transitive closure of $R$. 
      
      \begin{algorithm}
        \caption{Determine whether $K,w_0 \models \varphi$.}
        \label{alg:ctl-mc-upper-bound-logcfl}
        \begin{algorithmic}[1]
          \REQUIRE a Kripke structure $K=(W,R,\eta)$, $w_0 \in W$, $\varphi \in \CTLpos(T)$
          \STATE $\mathrm{push}(S,(\varphi,w_0))$
          \WHILE{$S$ is not empty}
            \STATE $(\varphi,w) \leftarrow \mathrm{pop}(S)$
            \IF{$\varphi$ is a propositional formula} \label{alg:ctl-mc-upper-bound-logcfl-if-start}
              \IF{$\varphi$ evaluates to false in $w$ under $\eta$}
              \RETURN \FALSE
            \ENDIF
            \ELSIF{$\varphi = \alpha \land \beta$}
              \STATE $\mathrm{push}(S,(\beta,w))$
              \STATE $\mathrm{push}(S,(\alpha,w))$
            \ELSIF{$\varphi = \alpha \lor \beta$}
              \STATE nondet.\ $\mathrm{push}(S,(\alpha,w))$ or $\mathrm{push}(S,(\beta,w))$
            \ELSIF{$\varphi = \EX\alpha$}
              \STATE nondet.\ choose $w' \in \{ w' \mid (w,w') \in R\}$ \label{alg:ctl-mc-upper-bound-logcfl-EX}
              \STATE $\mathrm{push}(S,(\alpha,w'))$
            \ELSIF{$\varphi = \EF\alpha$}
              \STATE nondet.\ choose $w' \in \{ w' \mid (w,w') \in R^\star\}$ \label{alg:ctl-mc-upper-bound-logcfl-EF}
              \STATE $\mathrm{push}(S,(\alpha,w'))$ 
            \ENDIF \label{alg:ctl-mc-upper-bound-logcfl-if-end}
          \ENDWHILE
          \RETURN \TRUE
        \end{algorithmic}
      \end{algorithm}

      Algorithm~1 always terminates
      because each subformula of $\varphi$ is pushed to the stack $S$ at most once.
      For correctness, an induction on the structure of formulae shows
      that Algorithm~1 returns
      \textbf{false} if and only if for the most recently popped pair
      $(\psi,w)$ from $S$, we have $K,w \not\models \psi$. Thence, in
      particular, Algorithm~1
      returns \textbf{true} iff $K,w \models \varphi$. 

      Algorithm~1 can be implemented on a
      nondeterministic polynomial-time Turing machine that besides its (unbounded) stack uses only logarithmic memory for the local variables.
      Thus $\CTLposMC(T)$ is in $\LOGCFL$. 
      
      The case $T \subseteq \{\AX,\AG\}$ is analogous and 
      follows from closure of $\LOGCFL$ under complementation.
    \end{proof}
    
    Finally, for the trivial case where no $\CTL$-operators are present, 
    model checking $\CTL(\emptyset)$-formulae is equivalent to the problem 
    of evaluating a propositional formula. 
    This problem is known to be solvable in $\NC1$ \cite{bus87}.

  \subsection{Lower Bounds}\label{subsect:lower bounds}

    The \P-hardness of model checking for \CTL was first stated in \cite{schnoeb02}.
    We improve this lower bound and concentrate on the smallest fragments of monotone \CTL---w.r.t.\ \CTL-operators---with \P-hard model checking.
    
    \begin{prop}\label{prop:ctlmon-mc-P-hardness}
      Let $T$ denote a set of $\CTL$-operators. 
      Then $\CTLmonMC(T)$ is $\P$-hard if $T$ contains an existential and a universal $\CTL$-operator.
    \end{prop}

    \begin{proof}
      First, assume that $T=\{\AX,\EX\}$.
      We give a generic reduction from the word problem for alternating Turing machines working in logarithmic space,
      which follows the same line as the classical proof idea (see \cite[Theorem 3.8]{schnoeb02}),
      and which we will modify in order to be useful for other combinations of \CTL-operators.
      Let $M$ be an alternating logspace Turing machine, and let $x$ be an input to $M$.
      We may assume w.l.o.g.\ that each transition of $M$ leads from an existential to a universal configuration and vice versa. 
      Further we may assume that each computation of $M$ ends after the same number $p(n)$ of steps, where $p$ is a polynomial and $n$ is the length of $M$'s input.
      Furthermore we may assume that there exists a polynomial $q$ such that $q(n)$ is the number of configurations of $M$ on any input of length $n$.
      
      Let $c_1,\ldots,c_{q(n)}$ be an enumeration of all possible configurations of $M$ on input $x$, starting with the initial configuration $c_1$.
      We construct a Kripke structure $K:=(W,R,\eta)$ by defining the set $W:=\{c_i^j \mid 1 \leq i \leq q(n), 0 \leq j \leq p(n)\}$ and the relation $R \subseteq W \times W$ as
      \begin{align*}
        R  :=   &\; \big\{ (c_i^j,c_k^{j+1}) \;\big|\; M \text{ reaches configuration } c_k \text{ from } c_i \text{ in one step}, 0 \leq j < p(n) \big\} \\
           \cup &\;
           \big\{ (c_i^{j},c_i^{j}) \;\big|\; c_i^j\text{ has no successor}, 1 \leq i \leq q(n),0\leq j< p(n)\big\}\\
           \cup &\; \big\{ (c_i^{p(n)},c_i^{p(n)}) \;\big|\; 1 \leq i \leq q(n)\big\}.
      \end{align*}
      The labelling function $\eta$ is defined for all $c_i^j\in W$ as
      \[
        \eta(c_i^j) := \left\{\begin{array}{cl}
                             \{t\}, & \text{if $c_i$ is an accepting configuration and $j=p(n)$} \\
                             \emptyset, & \text{otherwise}
                          \end{array}
                   \right.
      \]
      where $t$ is the only atom used by this labelling.
      It then holds that
      \[
        M \text{ accepts } x \iff K,c_1^0 \models \psi_1\Big(\psi_2\big(\cdots \psi_{p(n)}(t)\big)\cdots\Big),
      \]
      where $\psi_i(x):=\AX(x)$ if $M$'s configurations \emph{before} the $i$th step are universal, and $\psi_i(x):=\EX(x)$ otherwise.
      Notice that the constructed $\CTL$-formula does not contain any Boolean connective.
      Since $p(n)$ and $q(n)$ are polynomials, the size of $K$ and $\varphi$ is polynomial in the size of $(M,x)$. 
      Moreover, $K$ and $\varphi$ can be constructed from $M$ and $x$ using $\AC{0}$-circuits. 
      Thus, $A \leqcd \CTLmonMC(\{\AX,\EX\})$ for all $A \in \ALOGSPACE=\P$.
      
      For $T=\{\AF,\EG\}$ we modify the above reduction by defining the labelling function $\eta$ and the formula $\psi_i$ as follows: 
      \begin{equation}\label{eq:ctl-mc-P-hardness-1}
         \begin{array}{@{}r@{\,}c@{\,}l@{}}
         \eta(c_i^j) & := & \left\{ 
           \renewcommand{\arraystretch}{1}
           \begin{array}{@{}l@{~~}l@{}}
             \{d_j, t\}, & \text{if $c_i$ is an accepting configuration and $j=p(n)$} \\
             \{d_j \}, & \text{otherwise}
           \end{array}
           \right. \\[2ex]
         
         \psi_i(x) & := &
         \left\{
         \renewcommand{\arraystretch}{1}
         \begin{array}{@{}l@{~~}l@{}}
           \AF(d_{i} \land x), & \text{if $M$'s configurations \emph{before} step $i$ are universal, }\\[2.5pt]
           \EG(D_{i} \lor x), & \text{otherwise}, \\
         \end{array}
         \right.
         \end{array}
      \end{equation}
      where $d_j$ are atomic propositions encoding the `time stamps' of the respective configurations and $D_i = \bigvee_{i \neq j \in \{0,\ldots,p(n)\}} d_j$. 

      For the combinations of $T$ being one of
      $\{\AF, \EF\}$,
      $\{\AF, \EX\}$,
      $\{\AG, \EG\}$, 
      $\{\AG,$ $\EX\}$, 
      $\{\AX, \EF\}$, 
      and $\{\AX,\EG\}$, 
      the $\P$-hardness of 
      $\CTLmonMC(T)$ is obtained using analogous modifications to $\eta$ and the $\psi_i$'s.
      
      For the remaining combinations involving the until or the release operator, observe that w.r.t.\
      the Kripke structure $K$ as defined in \eqref{eq:ctl-mc-P-hardness-1}, 
      $\AF(d_i \land x)$ and $\EG(D_i \lor x)$ are equivalent to $\A[d_{i-1} \U x]$ and $\E[d_{i-1} \U x]$,
      and that $\R$ and $\U$ are duals.
    \end{proof}

    In the presence of arbitrary negation, universal operators are definable by
    existential operators and vice versa. Hence, from
    Proposition~\ref{prop:ctlmon-mc-P-hardness} we 
    obtain the following corollary.    
    \begin{cor}  \label{cor:ctl-mc-P-hardness}
      The model checking problem $\CTLMC(T)$ is $\P$-hard for each
      nonempty set $T$ of $\CTL$-operators. 
    \end{cor}
    
    Returning to monotone \CTL, in
    most cases even one operator suffices to make model checking
    $\P$-hard:    
    \begin{prop} \label{prop:ctl-mon-mc-EG-is-P-hard}
      Let $T$ denote a set of $\CTL$-operators. 
      Then $\CTLmonMC(T)$ is $\P$-hard if $T$ contains at least one of the operators $\EG$, $\EU$, $\ER$, $\AF$, $\AU$, or $\AR$.
    \end{prop}


    \begin{proof}
      We modify the proof of Proposition~\ref{prop:ctlmon-mc-P-hardness} to work with $\EG$ only. 
      The remaining fragments follow from the closure of $\P$ under complementation and 
      $\F \chi \equiv \neg \G \neg \chi \equiv [\true \U  \chi]$, $[\chi \U \pi] \equiv \neg [\neg \chi \R \neg\pi]$.
      
      Let the machine $M$, the word $x$, the polynomials $p, q$, and $K$ be as above.
      Further assume w.l.o.g.\ that $M$ branches only binary in each step. 
      Denote by $W_\exists$ (resp.\ $W_\forall$) the set of states corresponding to existential (resp.\ universal) configurations. The purpose of the introduced layers below is to ensure the uniqueness of the successors of universal configurations which is essential in the construction of $\psi_i$ later.
      We construct a Kripke structure $K':=(W',R,\eta)$ consisting of $q(n)+1$ layers and a `trap' as follows:
      let $W':= W \times \{1,\ldots,q(n) + 1\} \cup \{z\}$.
      The transition relation $R \subseteq W' \times W'$ is defined as
      \newlength{\temp}
      \settowidth{\temp}{$\big((c_k^{j+1},1),(c_k^{j+1},2)\big),$}
      \[
      \begin{array}{@{}r@{\,}c@{\,}l@{}}
        R&:=&\left\{
               \big((c_i^j,\ell),(c_k^{j+1},\ell)\big)
             \;\middle|\;
             \begin{array}{@{}l@{}}
               c_i^j \in W_\exists, M \text{ reaches } c_k \text{ from } c_i \text{ in one step}, \\
               1 \leq \ell \leq q(n)+1, 0 \leq j < p(n)
             \end{array}
        \right\} \\
         &\cup&\left\{
          \begin{array}{@{}l@{}}
            \big((c_i^j,\ell),(c_k^{j+1},i)\big), \\ 
            \big((c_k^{j+1},i),(c_{k'}^{j+1},q(n)+1)\big), \\ 
            \big((c_{k'}^{j+1},q(n)+1),z\big)
          \end{array}
          \;\middle|\;
          \begin{array}{@{}l@{}}
            c_i^j \in W_\forall, M \text{ reaches } c_k \text{ and } c_{k'} \text{ from } c_i \text{ in} \\
            \text{one step}, c_k \leq c_{k'}, 0 \leq j < p(n)
          \end{array}
        \right\} \\
        &\cup& \big\{\big((c_i^{p(n)},\ell),(c_i^{p(n)},\ell)\big) \mid 1\leq i\leq q(n), 1\leq\ell\leq q(n)+1 \big\} \\
        &\cup& \big\{(z,z)\big\}.
      \end{array} 
      \]
      That is, the arcs leaving an existential configurations $c_i$ lead to the successor configurations of $c_i$ inside each layer;
      while any universal configuration $c_i$ has exactly one outgoing arc pointing to its (lexicographically) first successor 
      configuration in the layer $i$, from where another arc leads to the second successor of $c_i$ in layer $q(n)+1$, 
      which in turn has an outgoing arc to the state $z$ (see Figure~\ref{fig:K'}).
      \begin{figure}
        \includegraphics[width=0.6\textwidth]{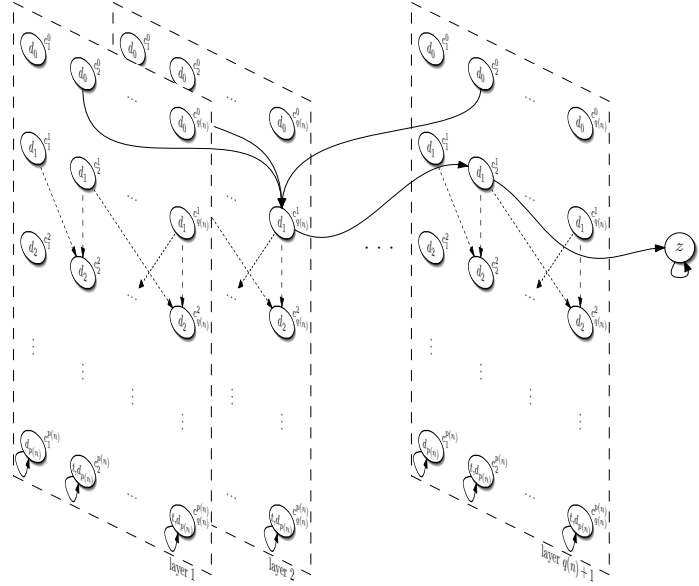}
        \caption{\label{fig:K'}%
          The Kripke structure $K'$; dashed (resp.\ solid) arrows correspond to transitions leaving existential (resp.\ universal) configurations.%
        }
      \end{figure}
      The labelling function $\eta$ is defined as 
        $\eta(z):= \{z\}$,
        $\eta((c_i^j,\ell)):= \{\ell,d_j,t\}$ if $c_i$ is an accepting configuration, and otherwise 
        $\eta((c_i^j,\ell)):= \{\ell,d_j\}$ for ($1 \leq \ell \leq q(n)+1$).
      Define 
      \[
        \psi_i(x) := 
        \begin{cases}
          \EG(d_{i-1} \lor (d_i\land x)  \lor z), \text{if $M$'s configurations before step $i$ are universal,} \\
          \EG(D_{i} \lor x), \text{if $M$'s configurations before step $i$ are existential,}\\
        \end{cases}
      \]
      and $D_i = \bigvee_{i \neq j \in \{0,\ldots,p(n)\}} d_j$. 
      The correctness of the equivalence
      $K,w \models \AF(d_i \land x)$ iff 
      $K',(w,\ell) \models \EG\big(d_{i-1} \lor (d_i\land x)  \lor z)\big)$,
      for all $w \in W_\forall$, $1 \leq \ell \leq q(n)+1$ and $1 \leq i \leq p(n)$ can be verified through the following observations. $\Rightarrow$: if $d_i$ and $x$ hold in all successors of $w$ in $K$, then there exists a path from $(w,\ell)$ to both of the series-connected successors reaching the trap and looping there. This is the only possibility for the path as neither $d_{i-1}$ nor $d_i$ hold below that level. As in each successor configuration the subformula $d_i\land x$ must be satisfied the composition of the $\psi_i$s ensures that in each such state there must start an $\EG$-path for each universal successor. $\Leftarrow$: the only path which satisfies at least one of the three disjuncts ranges through both series-connected successor configurations and ends in the trap. For each of the two successor states $d_i$ and $x$ hold. Thus $\AF(d_i\land x)$ is true in the state $w$ in the structure $K$.
      
      From this, it easily follows that for 
      \[
        M \text{ accepts } x  \iff  K',(c_1^0,1)   \models   \psi_1\Big(\psi_2\big( \cdots \psi_{p(n)}(t)\big)\cdots \Big).
      \]
      As we essentially only duplicated the set of states in $K$ and 
      $R$ can be constructed from all triples of states in $W'$, 
      $K'$ remains $\AC{0}$ constructible.
      Concluding $A \leqcd \CTLmonMC(\{\EG\})$ for all $A \in \P$. 
    \end{proof}

    By Lemma~\ref{prop:simple-cdreds-mc},
    $\CTLmonMC(T)\leqcd\CTLposMC(T)$ and hence
    the above results directly translate to model checking for $\CTLpos$:
    for any set $T$ of temporal operators, 
    $\CTLposMC(T)$ is $\P$-hard if $T \nsubseteq \{\EX,\EF\}$ or if $T \nsubseteq \{\AX,\AG\}$.
    These results cannot be improved w.r.t.\ $T$, as for $T \subseteq \{\EX,\EF\}$ and $T \subseteq \{\AX,\AG\}$
    we obtain a \LOGCFL upper bound for model checking from Proposition~\ref{prop:LOGCFL_upper_bound}.
    In the following proposition we prove the matching \LOGCFL lower bound.

    \begin{prop}\label{prop:LOGCFL-hard}
      For every nonempty set $T$ of $\CTL$-operators,
      the model checking problem $\CTLmonMC(T)$ is $\LOGCFL$-hard.
    \end{prop}
  
    \begin{proof}
      As explained in Section~\ref{subsect:complexity}, 
      $\LOGCFL$ can be characterized as the set of languages recognizable by logtime-uniform $\SAC1$ circuits, 
      \emph{i.e.}, circuits of logarithmic depth and polynomial size 
      consisting of $\lor$-gates with unbounded fan-in and $\land$-gates with fan-in $2$. 
      For every single $\CTL$-operator $O$,
      we will show that $\CTLmonMC(T)$ is $\LOGCFL$-hard for all $T\supseteq\{O\}$
      by giving a generic $\leqcd$-reduction $f$ from the word problem for $\SAC1$ circuits to $\CTLmonMC(T)$. 
      
      First, consider $\EX \in T$.
      Let $C$ be a logtime-uniform $\SAC1$ circuit of depth $\ell$ with $n$ inputs and let $x=x_1\dots x_n \in \{0,1\}^n$. 
      Assume w.l.o.g.\ 
        that $C$ is connected, 
        layered into alternating layers of $\land$-gates and $\lor$-gates, and
        that the output gate of $C$ is an $\lor$-gate.
      We number the layers bottom-up, that is, the layer containing (only) the output gate has level $0$, whereas the input-gates and negations of the input-gates are situated in layer $\ell$. Denote the graph of $C$ by $G=(V,E)$, where $V:=V_\mathrm{in} \uplus V_\land \uplus V_\lor$ is partitioned into the sets corresponding to the (possibly negated) input-gates, the $\land$-gates, and the $\lor$-gates, respectively.
      $G$ is acyclic and directed with paths leading from the input to the output gates.
      From $(V,E)$ we construct a Kripke structure that allows to distinguish the two predecessors of an $\land$-gate from each other. 
      This will be required to model proof trees using $\CTLmon(\{\EX\})$-formulae. 
      
      For $i \in \{1,2\}$, let $V_{\mathrm{in}}^i:=\{v^i \mid v \in V_{\mathrm{in}}\}$, 
          $V_{\lor}^i:=\{v^i \mid v \in V_{\lor}\}$ and define $V_{\mathrm{in},\lor}^i:= V_\mathrm{in}^i \cup V_\lor^i$.
      Further define 
      \begin{align*}
        E':= \;
        & \big\{(v,u^i) \in V_\land \times V_{\mathrm{in},\lor}^i \mid (u,v) \in E \text{ and $u$ is the $i$th predecessor of $v$} \big\} \\
        \cup \;
        & \big\{(v,v) \mid v \in V_{\mathrm{in}}^1 \cup V_{\mathrm{in}}^2 \big\} \cup \bigcup_{i \in \{1,2\}} \big\{(v^i,u) \in V_{\mathrm{in},\lor}^i \times  V_\land \mid (u,v) \in E \big\},
      \end{align*}
      where the ordering of the predecessors is implicitly given in the encoding of $C$.
      We now define a Kripke structure $K:=(V',E',\eta)$ with 
       states $V':=V_{\mathrm{in},\lor}^1 \cup V_{\mathrm{in},\lor}^2 \cup V_\land$, transition relation $E'$, and
       labelling function $\eta:V'\rightarrow\powerset{\{1,2,t\}}$, 
      \[
        \eta(v) :=
        \begin{cases}
          \{i,t\},   & \text{if } (v=v_{\mathrm{in}_j} \in V_\mathrm{in}^i \text{ and } x_j=1) \text{ or } (v=\overline{v}_{\mathrm{in}_j} \in V_\mathrm{in}^i \text{ and } x_j=0), \\
          \{i\},   & \text{if } (v=v_{\mathrm{in}_j} \in V_\mathrm{in}^i \text{ and } x_j=0) \text{ or } (v=\overline{v}_{\mathrm{in}_j} \in V_\mathrm{in}^i \text{ and } x_j=1) \text{ or } v \in V_\lor^i, \\
          \emptyset, & \text{otherwise},
        \end{cases}
      \]
      where $i=1,2$, $j=1,\ldots,n$ and $v_{\mathrm{in}_1},\ldots, v_{\mathrm{in}_n}$, $\overline{v}_{\mathrm{in}_1},\ldots,\overline{v}_{\mathrm{in}_n}$
      enumerate the input gates and their negations.
      The formula $\varphi$ that is to be evaluated on $K$
      will consist of atomic propositions $1$, $2$ and $t$,
      Boolean connectives $\land$ and $\lor$, and the $\CTL$-operator $\EX$.
      To construct $\varphi$ we recursively define formulae 
      $(\varphi_i)_{0 \leq i \leq \ell}$ by
      \[
      \varphi_i :=
      \begin{cases}
        t, & \text{if $i=\ell$}, \\
        \EX \varphi_{i+1}, & \text{if $i$ is even ($\lor$-layers)}, \\
        \bigwedge_{i=1,2}\EX(i \land \varphi_{i+1}), & \text{if $i$ is odd ($\land$-layers)}.
      \end{cases}
      \]
      We define the reduction function $f$ as the mapping $(C,x) \mapsto (K,v_0,\varphi)$, 
      where $v_0$ is the node corresponding to the output gate of $C$ and $\varphi:=\varphi_0$.
      We stress that the size of $\varphi$ is polynomial, for the depth of $C$ is logarithmic only.
      Clearly, each minimal accepting subtree (cf. \cite{ruz80} or \cite[Definition 4.15]{vol99}) of $C$ on input $x$ translates into a sub-structure $K'$ of $K$ 
      such that $K',v_0 \models \varphi$, where
      \begin{enumerate}
        \item $K'$ includes $v_0$,
        \item $K'$ includes one successor for every node corresponding to an $\lor$-gate, and
        \item $K'$ includes the two successors of every node corresponding to an $\land$-gate.
      \end{enumerate}
      As $C(x)=1$ iff there exists a minimal accepting subtree of $C$ on $x$, the $\LOGCFL$-hardness of $\CTLmonMC(T)$ for $\EX \in T$ follows.
      
      Second, consider $\EF \in T$.
      We have to extend our Kripke structure to contain information about the depth of the corresponding gate. 
      We may assume w.l.o.g.\ that $C$ is encoded such that
      each gate contains an additional counter holding the distance to
      the output gate (which is equal to the number of the layer it is
      contained in, cf.~\cite{vol99}). 
      We extend $\eta$ to encode this distance $i$, $1\leq i \leq \ell$, into the ``depth-propositions'' $d_i$ as in the proof of Proposition~\ref{prop:ctlmon-mc-P-hardness}.
      Denote this modified Kripke structure by $K'$.
      Further, we define $(\varphi_i')_{0 \leq i \leq \ell}$ as
      \[
      \varphi_i' :=
      \begin{cases}
        t, & \text{if $i=\ell$}, \\
        \EF(d_{i+1} \land \varphi_{i+1}'), & \text{if $i$ is even}, \\
        \bigwedge_{i=1,2} \EF(d_{i+1} \land i \land \varphi_{i+1}'), & \text{if $i$ is odd}.
      \end{cases}
      \]
      Redefining the reduction $f$ as $(C,x) \mapsto (K',v_0,\varphi_0')$ hence yields the \LOGCFL-hardness of $\CTLmonMC(T)$ for $\EF \in T$.
      
      Third, consider $\AX \in T$.
      Consider the reduction in case 1 for $\CTLmon(\{\EX\})$-formulae,
      and let $f(C,x)=(K,v_0,\varphi)$ be the value computed by the reduction function.
      It holds that $C(x)=1$ iff $K,v_0\models\varphi$,
      and equivalently $C(x)=0$ iff $K,v_0\models\neg\varphi$.
      Let $\varphi'$ be the formula obtained from $\neg\varphi$ by multiplying
      the negation into the formula.
      Then $\varphi'$ is a $\CTLan(\{\AX\})$-formula.
      Since \LOGCFL is closed under complement,
      it follows that $\CTLanMC(\{\AX\})$ is \LOGCFL-hard.
      Using Lemma~\ref{lem:atomic-negation=no-negation},
      we obtain that $\CTLmonMC(\{\AX\})$ is \LOGCFL-hard, too.
      An analogous argument works for the case $\AG\in T$.
      The remaining fragments are even $\P$-complete by Proposition~\ref{prop:ctl-mon-mc-EG-is-P-hard}.
    \end{proof}

    Using Lemma~\ref{prop:simple-cdreds-mc} we obtain 
    $\LOGCFL$-hardness of $\CTLposMC(T)$ for all nonempty sets $T$ of $\CTL$-operators.


    In the absence of $\CTL$-operators, the lower bound for the model
    checking problem again follows from the lower bound for evaluating monotone propositional formulae.
    This problem is known to be hard for $\NC{1}$ \cite{bus87,sc10}.

%
%
%

  \subsection{The Power of Negation}\label{subsect:negation}

  We will now show that model checking for the fragments $\CTLan$ and $\CTLpos$
  is computationally equivalent to model checking for $\CTLmon$, 
  for any set $T$ of \CTL-operators.
  Since we consider $\leqcd$-reductions, this is not immediate.

  From Lemma~\ref{prop:simple-cdreds-mc}
  it follows that the hardness results for $\CTLmonMC(T)$
  also hold for $\CTLanMC(T)$ and $\CTLposMC(T)$.
  Moreover, the algorithms for $\CTLposMC(T)$ 
  also work for $\CTLmonMC(T)$ and $\CTLanMC(T)$ without using more computation resources.
  Both observations together yield
  the same completeness results for all \CTL-fragments with restricted negations.

  \begin{thm}\label{thm2_ext}
  Let $T$ be any set of \CTL-operators. 
  Then $\CTLmonMC(T)$, $\CTLanMC(T)$, and $\CTLposMC(T)$ are 

  \begin{iteMize}{$-$}
    \item
      \NC1-complete if $T$ is empty, 
    \item
      \LOGCFL-complete if $\emptyset\subsetneq T\subseteq\{\EX,\EF\}$ or 
      $\emptyset\subsetneq T \subseteq \{\AX,\AG\}$, 
    \item
      \P-complete otherwise.
  \end{iteMize}

  Moreover,  the problems $\CTLmonMC(T)$, $\CTLanMC(T)$, and $\CTLposMC(T)$
  are equivalent w.r.t.\ $\leqcd$-reductions.
  \end{thm}

  This equivalence extends Lemma~\ref{lem:atomic-negation=no-negation}.
  We remark that this equivalence is not straightforward. 
  Simply applying de Morgan's laws to transform one problem into another requires 
  counting the number of negations on top of $\land$- and $\lor$-connectives. 
  This counting cannot be achieved by an $\AC0$-circuit and does not lead to the aspired reduction. 
  Here we obtain equivalence of the problems as a consequence 
  of our generic hardness proofs in Section~\ref{subsect:lower bounds}.

\section{Model Checking Extensions of CTL}\label{sect:extensions}

It has been argued that $\CTL$ lacks the ability to express fairness properties.
To address this shortcoming, Emerson and Halpern introduced $\ECTL$ in \cite{emha86}.
$\ECTL$ extends $\CTL$ with the $\Finfty$-operator, which states that for 
every moment in the future, the enclosed formula will eventually be satisfied again:
for a Kripke structure $K$, a path $x=(x_1,x_2,\ldots)$, and a path formula $\chi$ 
\[
  K,x\models \Finfty\chi \ \ \text{ iff } \ \ K,x^i \models \F\chi\ \mbox{ for all } i\in \N.
\]
The dual operator $\Ginfty$ is defined analogously.
As for \CTL, model checking for $\ECTL$ is known to be tractable. 
Moreover, our next result shows that even for all fragments,
model checking for \ECTL is not harder than for \CTL.

\begin{thm} \label{thm:ectl-mc}
  Let $T$ be a set of temporal operators. 
  Then $\ECTLMC(T)\equivcd\CTLMC(T')$
  and $\ECTLposMC(T)\equivcd\CTLposMC(T')$,
  where $T'$ is obtained from $T$ by substituting $\Finfty$ with $\F$ and $\Ginfty$ with $\G$.
\end{thm}

\begin{proof}
  For the upper bounds, notice that $\ECTLMC(\ALL \cup \{\EFinfty, \AFinfty\}) \in \P$. 
  It thus remains to show that $\ECTLposMC(T) \in \LOGCFL$ for $T \subseteq \{\EX,\EF,\EFinfty\}$ and $T \subseteq \{\AX,\AG, \AGinfty\}$
  First, consider the case that $T \subseteq \{\EX, \EF, \EFinfty\}$.
  We modify Algorithm~1 to handle $\EFinfty$
  by extending the case distinction in lines \ref{alg:ctl-mc-upper-bound-logcfl-if-start}--\ref{alg:ctl-mc-upper-bound-logcfl-if-end} with the code fragment 
  given in Algorithm~1. 
  The algorithm for $T \subseteq \{\AX,\AG,\AGinfty\}$ is analogous and membership in $\LOGCFL$ follows from its closure under complementation.
  \begin{algorithm}[t]
    \caption{\label{alg:ectl-mc-upper-bound-logcfl}%
      Case distinction for $\EFinfty$}
    \begin{algorithmic}[1]
      \STATE \algorithmicelsif $\;\varphi = \EFinfty\alpha$ \algorithmicthen
      \makeatletter
      \begin{ALC@if}
      \STATE nondet. choose $k \leq \size{W}$ and a path $(w_i)_{1 \leq i \leq k}$ such that $(w,w_1)\in R^\star$, $(w_k,w_1) \in R$
        \STATE nondet. choose some $1 \leq i \leq k$ and $\mathrm{push}(S,(\alpha,w_i))$
      \ENDIF
      \makeatother
    \end{algorithmic}
  \end{algorithm}

  For the lower bounds, 
  we extend the proofs of Propositions~\ref{prop:ctlmon-mc-P-hardness}, \ref{prop:ctl-mon-mc-EG-is-P-hard} and \ref{prop:LOGCFL-hard} to handle sets $T$ 
  involving also the operators $\AFinfty$, $\AGinfty$, $\EFinfty$, and $\EGinfty$. 
  Therefore, we only need modify the accessibility relation $R$ of respective Kripke structure $K$ to be reflexive. 
  The hardness results follow by replacing $\F$ with $\Finfty$ and $\G$ with $\Ginfty$ in the respective reductions.
  
  First consider the case that $T$ contains an existential and a universal operator, say $T=\{\AFinfty,\EGinfty\}$.
  Let $M$, $x$, and $p$ be defined as in the proof of Proposition~\ref{prop:ctlmon-mc-P-hardness}. 
  We map $(M,x)$ to $(\tilde{K},c_1^0,\psi_1)$, where $\tilde{K}=(W,R,\eta)$ is the reflexive closure of the Kripke structure $K$ defined for the $\P$-hardness of 
  $\CTLMC(\{\AF,\EG\})$, $c_1^0 \in W$, and $\psi:=\psi_1\Big(\psi_2\big(\cdots \psi_{p(n)}(t)\big)\cdots\Big)$, where
  \[
        \begin{array}{@{}r@{\,}c@{\,}l@{}}
        \psi_i(x) & := &
        \left\{
        \renewcommand{\arraystretch}{1}
        \begin{array}{@{}l@{~~}l@{}}
          \AFinfty(d_i \land x), & \text{if $M$'s configurations in step $i$ are universal, }\\[2.5pt]
          \EGinfty(D_i \lor x), & \text{otherwise}, \\
        \end{array}
        \right.
        \end{array}
    \]  
  In $\tilde{K}$ it now holds that $d_i \in \eta(w)$ and $(w,w') \in R$ together imply that either $w = w'$ or $d_{i} \notin \eta(w')$.
  Hence, for all $w \in W$ and $1 \leq i \leq p(|x|)$, $\tilde{K},w \models \AFinfty(d_i \land x)$ iff $K,w \models \AF(d_i \land x)$, and  
  $\tilde{K},w \models \EGinfty(\bigvee_{i \neq j \in \{0,\ldots,p(n)\}} d_j \lor x)$ iff $K,w \models  \EG(\bigvee_{i \neq j \in \{0,\ldots,p(n)\}} d_j \lor x)$. 
  From this, correctness of the reduction follows.
  The $\P$-hardness of $\CTLMC(T)$ for the remaining fragments follows analogously. 
  
  As for $T \subseteq \{\EX,\EF,\EFinfty\}$,
  we will show that $\ECTLmonMC(T)$ is $\LOGCFL$-hard under $\leqcd$-reductions for $T=\{\EFinfty\}$. 
  Let $C$, $x$, and $\ell$ be as in the proof of Proposition~\ref{prop:LOGCFL-hard}. We map the pair
  $(C,x)$ to the triple $(\tilde{K'},v_0,\varphi_0)$, where $\tilde{K'}=(V',E',\eta)$ is the reflexive closure of the Kripke structure $K'$ 
  defined for the $\LOGCFL$-hardness of $\CTLMC(\{\EF\})$, $v_0 \in V'$, and $\varphi_0$ is recursively defined via $(\varphi_i')_{0 \leq i \leq \ell}$ as
  \[
  \varphi_i :=
  \begin{cases}
    t, & \text{if $i=\ell$}, \\
    \EFinfty(d_{i+1} \land \varphi_{i+1}), & \text{if $i$ is even}, \\
    \bigwedge_{i=1,2} \EFinfty(d_{i+1} \land i \land \varphi_{i+1}), & \text{if $i$ is odd}.
  \end{cases}
  \]
  Again, we have that in $\tilde{K'}$, $d_i \in \eta(v)$ and $(v,v') \in E'$ together imply that either $v = v'$ or $d_{i} \notin \eta(v')$.
  It hence follows $\tilde{K'},v \models \EFinfty(d_i \land \varphi_{i})$ iff $K',v \models \EF(d_i \land  \varphi_{i})$,  for all $v \in V'$ and $1 \leq i \leq \ell$. 
  We conclude that $\ECTLmonMC(\{\EFinfty\})$ is $\LOGCFL$-hard. The case $T=\{\AGinfty\}$ follows analogously. 
\end{proof}

We will now consider $\CTLplus$, the extension of \CTL by Boolean combinations of path formulae which is defined as follows. 
A \emph{$\CTLplus$-formula} is a $\CTLstar$-formula where each pure temporal operator in a state formula occurs in the scope 
of a path quantifier. The set of all $\CTL$-formulae is a strict subset of the set of all $\CTLplus$-formulae, 
which again forms a strict subset of the set of all $\CTLstar$-formulae.
For example, $\AG\EF p$ and $\A(\G p \land \F q)$ are $\CTLplus$-formulae, but $\A\G\F p$ is not. 
However, $\CTL$ is as expressive as $\CTLplus$ \cite{emha85}.

By $\CTLplus(T)$ we denote the set of $\CTLplus$-formulae using the connectives $\{\land,\lor,\neg\}$ and
temporal operators in $T$ only. 
Analogous to the fragments $\CTLpos(T)$, $\CTLan(T)$, and $\CTLmon(T)$, 
we define $\CTLpluspos(T)$, $\CTLplusan(T)$, and $\CTLplusmon(T)$ as those fragments of $\CTLplus(T)$ that 
disallow temporal operators in the scope of negations,
contain negation signs only directly in front of atomic propositions, and
do not contain negation signs at all, respectively.

In contrast to \CTL, model checking for $\CTLplus$ 
is not tractable, but $\DeltaPtwo$-complete \cite{lamasc01}. 
Below we classify the complexity of model checking for both the full and the positive fragments of \CTLplus.

\begin{thm} \label{thm:ctl+-mc}
Let $T$ be a set of temporal operators containing at least one path quantifier. 
Then $\CTLplusMC(T)$ is 
  \begin{iteMize}{$-$}
    \item $\NC1$-complete if $T \subseteq \{\A,\E\}$, 
    \item $\P$-complete if $\{\X\} \subsetneq T \subseteq
      \{\A,\E,\X\}$, and 
    \item $\DeltaPtwo$-complete otherwise.
  \end{iteMize}
\end{thm}

\begin{proof}
  If $T \subseteq \{\A,\E\}$ then deciding $\CTLplusMC(T)$ is equivalent to the problem of evaluating a propositional formula, which is known to be $\NC1$-complete \cite{bus87,sc10}.
  
  
  If $\{\X\} \subsetneq T \subseteq \{\A,\E,\X\}$, then $\CTLplusMC(T)$ can be solved using a labelling algorithm:
  Let $K=(W,R,\eta)$ be a Kripke structure, and $\varphi$ be a $\CTLplus(\{\A,\E,\X\})$-formula.
  Assume w.l.o.g.\ that $\varphi$ starts with an $\E$ and that it does not contain any $\A$'s.
  Compute $K,w \models \psi$ for all $w \in W$ and all subformulae $\E\psi$ of $\varphi$ such that $\psi$ is free of path quantifiers, 
  and replace $\E\psi$ in $\varphi$ with a new proposition $p_\psi$ while extending the labelling function $\eta$ such that $p_\psi \in \eta(w) \iff K,w \models \psi$. 
  Repeat this step until $\varphi$ is free of path quantifiers and denote the resulting (propositional) formula by $\varphi'$.
  To decide whether $K,w \models \varphi$ for some $w \in W$, it now suffices to check whether $\varphi'$ is satisfied by the assignment implied by $\eta(w)$.
  As for all of the above subformulae $\E\psi$ of $\varphi$, $\psi \in \CTLplus(\{\X\})$, it follows that $K,w \models \psi$ can be determined in polynomial time in the size of $K$ and $\psi$. 
  Considering that the number of labelling steps is at most $O(|\varphi|\cdot |W|)$ it follows that $\CTLplusMC(T)$ is in $\P$. 
  The $\P$-hardness follows from $\CTLMC(\{\EX\}) \leqcd \CTLplusMC(\{\E,\X\})$ resp.\ $\CTLMC(\{\AX\}) \leqcd \CTLplusMC(\{\A,\X\})$.
  
  For all other possible sets $T$, we have $T \cap \{\E,\A\} \neq \emptyset$ and $T \cap \{\F,\G,\U\} \neq \emptyset$. 
  Consequently, each of the temporal operators $\A$, $\E$, $\F$, and $\G$ can be expressed in $\CTLplus(T)$. 
  The claim now follows from \cite{lamasc01}.
\end{proof}

For the positive fragments of \CTLplus we obtain a more complex picture:

\begin{thm}  \label{thm:ctl+pos-mc}
Let $T$ be a set of temporal operators containing at least one path quantifier. 
Then $\CTLplusposMC(T)$ is 
  \begin{iteMize}{$-$}
    \item $\NC1$-complete if $T \subseteq \{\A,\E\}$, 
    \item $\LOGCFL$-complete if $T = \{\A,\X\}$ or $T=\{\E,\X\}$,
    \item $\P$-complete if $T = \{\A,\E,\X\}$,
    \item $\NP$-complete if $\E \in T$, $\A\not\in T$ and $T$ contains a pure
      temporal operator aside from $\X$,
    \item $\co\NP$-complete if $\A \in T$, $\E\not\in T$ and $T$ contains a pure
      temporal operator aside from $\X$, and
    \item $\DeltaPtwo$-complete otherwise.
  \end{iteMize}
\end{thm}

\begin{proof}
  The first and third claim follow from Theorem~\ref{thm:ctl+-mc} and the monotone formula value problem being $\NC1$-complete \cite{sc10}.
  
  For the second claim, consider the case $T=\{\E,\X\}$.
  It is straightforward to adopt Algorithm~1 to guess a successor $w'$ of the current state once for every path quantifier $\E$ 
  that has been read and decompose the formula w.r.t.\ $w'$. For $T=\{\A,\X\}$ analogous arguments hold.
  
  The fourth claim can be solved with a labelling algorithm analogously to the algorithm for $\CTLplusMC(\{\A,\E,\X\})$.
  In this case, however, whole paths need to be guessed in the Kripke structures.
  Hence, we obtain a polynomial time algorithm deciding $\CTLplusposMC(T)$ using an oracle $B \in \NP$ (resp. $B \in \coNP$) .
  This algorithm is furthermore a monotone $\leqpt$-reduction from $\CTLplusposMC(T)$ to $B$, 
  in the sense that for any deterministic oracle Turing machine $M$ that executes the algorithm, 
  \[
    A \subseteq B  \implies L(M,A) \subseteq L(M,B),
  \]
  where $L(M,X)$ is the language recognized by $M$ with oracle $X$.
  Both $\NP$ and $\co\NP$ are closed under monotone $\leqpt$-reductions \cite{selman82}.
  We thus conclude that $\CTLplusposMC(T) \in \NP$ (resp.\ $\CTLplusposMC(T) \in \coNP$).
  
  As for the $\NP$-hardness of $\CTLplusposMC(T)$, note that the reduction from $\ThreeSAT$ to $\LTLMC(\{\F\})$, 
  the model checking problem for linear temporal logic using the $\F$-operator only, given by Sistla and Clarke in \cite{sicl85} is a reduction to $\CTLplusposMC(\{\E,\F\})$ indeed.
  The $\NP$-hardness of $\CTLplusposMC(\{\E,\G\})$ is obtained by a similar reduction:
  let $\varphi$ be a propositional formula in $\ThreeCNF$, \emph{i.e.}, 
  $\varphi=\bigwedge_{i=1}^n C_i$ with $C_i=\ell_{i1}\lor \ell_{i2} \lor \ell_{i3}$ and 
  $\ell_{ij}=x_k$ or $\ell_{ij}=\neg x_k$ for all $1 \leq i \leq n$, all $1 \leq j \leq 3$, and some $1 \leq k \leq m$. 
  Recall that for a set $A$, $\bigvee A$ denotes the disjunction $\bigvee_{a \in A} a$.
  We map $\varphi$ to the triple $(K,y_0,\psi)$, where $K=(W,R,\eta)$ is the Kripke structure given in \eqref{eq:ctlpluspos-mc-1} and 
  $\psi := \E \bigwedge_{i=1}^n \bigvee_{j=1}^3 \G \bigvee (\Phi\setminus\{\mathord\sim \ell_{ij}\})$ with 
  $\Phi := \{y_0,y_i,x_i,\overline{x}_i \mid 1 \leq i \leq m\}$
  and $\mathord\sim \ell_{ij}$ denoting the complementary literal of $\ell_{ij}$.
  \begin{align} \label{eq:ctlpluspos-mc-1}
    \begin{split}
      W &:= \{ y_0 \} \cup \{ x_i, \overline{x}_i, y_i \mid 1 \leq i \leq m\}, \\
      R &:= \{ (y_{i-1},x_i),(x_i,y_i),(y_{i-1},\overline{x}_i),(\overline{x}_i,y_i) \mid 1 \leq i \leq m \} \cup \{(y_m,y_m)\}, \\
      \eta(w) &:= \{w\} \text{ for all } w \in W.
    \end{split}
  \end{align}
  Note that the above reductions prove hardness for $\CTLplusmonMC(T)$ already.
  The $\co\NP$-hardness of $\CTLplusposMC(\{\A,\G\})$ and $\CTLplusposMC(\{\A,\F\})$ follows from the same reductions.
  
  As for the the last claim, note that the $\DeltaPtwo$-hardness of $\CTLplusMC(\{\A,\E,\F,\G\})$ carries over to $\CTLplusmonMC(\{\A,\E,\F,\G\})$, 
  because any $\CTLplus(\{\A,\E,\F,\G\})$-formula can be transformed into a $\CTLplusan(\{\A,\E,\F,\G\})$-formula, 
  in which all negated atoms $\neg p$ may be replaced by fresh propositions $\overline{p}$ that are mapped into all states of the Kripke structure whose label does not contain $p$.
  It thus remains to prove the $\DeltaPtwo$-hardness of $\CTLplusposMC(\{\A,\E,\F\})$ and $\CTLplusposMC(\{\A,\E,\G\})$.
  Consider $\CTLplusposMC(\{\A,\E,\G\})$. 
  Laroussinie \emph{et al.{}}\ reduce from $\SNSAT$, that is the problem to decide, 
  given disjoint sets $Z_1,\ldots,Z_n$ of propositional variables from $\{z_1,\ldots,z_p\}$ and a list 
  $
    \varphi_1(Z_1),
    \varphi_2(x_1,Z_2),
    \ldots,
    \varphi_n(x_1,\ldots,x_n,Z_n)
  $
  of formulae in conjunctive normal form, whether $x_n$ holds in the unique valuation $\sigma$ defined by 
  \begin{equation} \label{eq:ctl+pos-mc-1}
    \sigma(x_i) = \true  \iff \varphi_i(x_1,\ldots,x_{i-1},Z_i) \text{ is satisfiable.}
  \end{equation}
  An instance $I$ of $\SNSAT$ is transformed to the Kripke structure $K$ depicted in Figure~\ref{fig:ctl+posmc-DeltaP2-hard-1}
  \begin{figure}[t]
    \centering
    \includegraphics[width=\textwidth]{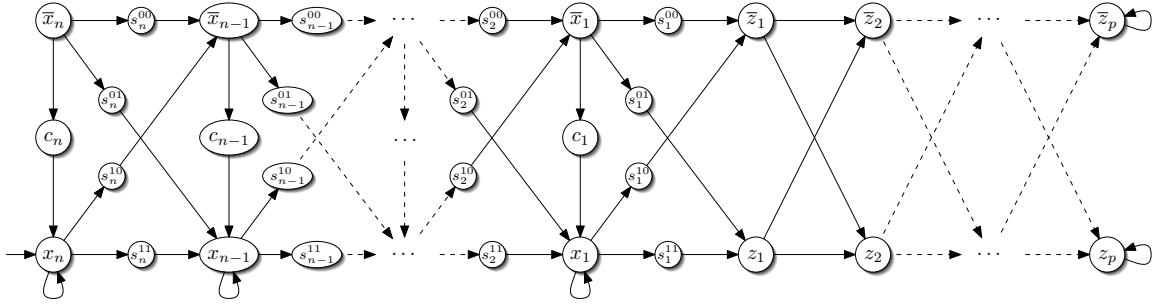}
    \caption{\label{fig:ctl+posmc-DeltaP2-hard-1}%
      Extended version of the Kripke structure constructed in~\cite[Figure~3]{lamasc01}.
    }
  \end{figure}
  and the formula $\psi_{2n-1}$ that is recursively defined as
  \begin{align*}
    \psi_k := \E \bigg[ \underbrace{\G\Big( \bigvee_{i=1}^n \overline{x}_i \limplies 
                        \E\big(\neg\F\bigvee_{i=1}^n (s_i^{00} \lor s_i^{01} \lor s_i^{10} \lor s_i^{11}) \land \F(\bigvee_{i=1}^n x_i \not\limplies \psi_{k-1})\big) \Big)}_{(A)} \hskip4em \\ 
                        \land \ \underbrace{\vphantom{\bigwedge_i}\G\Big( \bigwedge_{i=1}^n \neg c_i\Big)}_{(B)} \ 
                        \land \ \underbrace{\bigwedge_{i=1}^n \Big(\F x_i \limplies \bigwedge_j \bigvee_m \F \ell_{i,j,m} \Big)}_{(C)}
                 \bigg],
  \end{align*}
  for $1 \leq k \leq n$, $\psi_0:= \true$, and $\varphi_i = \bigwedge_j \bigvee_m \ell_{i,j,m}$, where the $\ell_{i,j,m}$'s are literals over $\{x_1,\ldots,x_n\} \cup Z_i$. Note that the structure $K$ from Figure~\ref{fig:ctl+posmc-DeltaP2-hard-1} differs from the Kripke structure constructed in \cite{lamasc01} in that we introduce different labels $c_i$ and $s_i^j$ for $1 \leq i \leq n$ and $j \in \{00,01,10,11\}$, as we need to distinguish between the states later on.
  The intuitive interpretation of $(B)$ is that the existentially quantified path does actually encode an assignment of $\{x_1,\ldots,x_n\}$ to $\{\false,\true\}$, 
  while $(C)$ states that this assignment coincides with $\sigma$ on all propositions that are set to $\true$. 
  Lastly $(A)$ expresses the recursion inherent in the definition of $\SNSAT$.
%
%
  It holds that $I \in \SNSAT \iff K,x_n \models \psi_{2n-1}$ (see \cite{lamasc01} for the correctness of this argument).
  
  We modify the given reduction to not use $\F$.
  First note that $\psi_{k-1}$ occurs negatively in $\psi_k$.
  We will therefore consider the formulae $\psi_{2n-1},\psi_{2n-3},\ldots,\psi_1$ and $\neg \psi_{2n-2},\neg\psi_{2n-4},\ldots,\neg\psi_2$ separately. In $\psi_{2n-1},\psi_{2n-3},\ldots,\psi_1$ replace
  \begin{iteMize}{$-$}
    \item $(A)$ with $\displaystyle\G\Big( \bigvee_{i=1}^n \overline{x}_i \limplies \E\big(\G\bigwedge_{i=1}^n(\neg s_i^{00} \land \neg s_i^{01} \land \neg s_i^{10} \land \neg s_i^{11}) \land \G(\bigvee_{i=1}^n \overline{x}_i \lor \bigvee_{i=1}^n c_i \lor \neg \psi_{k-1})\big)\Big)$,
    \item $(C)$ with $\displaystyle\bigwedge_{i=1}^n \Big(\G \neg x_i \lor \bigwedge_j \bigvee_m \G \bigvee (\Phi\setminus\{\mathord\sim\ell_{i,j,m}\}) \Big)$;
  \end{iteMize}
  and in $\neg \psi_{2n-2},\neg\psi_{2n-4},\ldots,\neg\psi_2$ replace
  \begin{iteMize}{$-$}
    \item $(A)$ with $\displaystyle\bigvee_{1 \leq i \leq n}\G\Big( \bigvee (\Phi\setminus\{\overline{x}_i\}) \lor \A\big( \G (\bigvee\Phi\setminus\{c_i\}) \lor \G(c_i \lor \psi_{k-1}\big) \Big)$,
    \item $(B)$ with $\displaystyle\bigvee_{i=2}^n\G\bigvee(\Phi \setminus \{s_i^{00},s_i^{01},s_{i-1}^{01},s_{i-1}^{11}\})$, 
    and
    \item $(C)$ with $\displaystyle\bigvee_{i=1}^n \Big(\G\bigvee(\Phi\setminus\{\overline{x}_i\}) \land \bigvee_j \bigwedge_m \G \neg \ell_{i,j,m} \Big)$,
  \end{iteMize}
  where $\Phi := \{x_i,\overline{x}_i,c_i,s_i^{00},s_i^{01},s_i^{10},s_i^{11} \mid 1\leq i \leq n\} \cup \{ z_i, \overline{z}_i \mid 1 \leq i \leq p \}$ is the set of all propositions used in $K$.
  Denote the resulting formulae by $\psi_k'$, $k\geq 0$.
  In $\psi_k'$, all negations are atomic and only the temporal operators $\E,\A$ and $\G$ are used. 
  
  To verify that $K,x_k \models \psi_k \iff K,x_k \models \psi_k'$ for all $0 \leq k < 2n$, consider $\psi_{k}$ with $k$ odd first. 
  Suppose $K,x_k \models \psi_k$. Then, by $(A)$, there exists a path $\pi$ in $K$ such that whenever some $\overline{x}_i$ is labelled in the current state $\pi_p$, then there exists a path $\pi'$ starting in $\pi_p$ that never visits any state labelled with $s_i^{j}$, $1 \leq i \leq n$, $j \in \{00,01,10,11\}$, and eventually falsifies $\psi_{k-1}$ because it reaches a state where neither $\overline x_i$ nor $c_i$ holds for all $1\leq i \leq n$.  
  Hence, by construction of $K$, $\pi'$ has to visit the states labelled with $c_i$ and $x_i$ for $i$ such that $\overline{x}_i \in \eta(\pi_p)$.
  This is equivalent to the existence of a path $\pi'$ starting in $\pi_p$ which never visits any state labelled with $s_i^{j}$, $1 \leq i \leq n$, $j \in \{00,01,10,11\}$, and that falsifies $\psi_{k-1}$ if the current state is not labelled with $c_i$ or $\overline{x}_i$ for all $1 \leq i \leq n$. 
  Hence the substitution performed on $(A)$ does not alter the set of states in $K$ on which the formula is satisfied. 
  
  The formula $(C)$, on the other hand, states that whenever the path $\pi$ quantified by the outmost $\E$ in $\psi_k$ visits the state labelled $x_i$, then for every clause $j$ in the $i$th formula $\varphi_i$ of given $\SNSAT$ instance at least one literal $\ell_{i,j,m}$ occurs in the labels on $\pi$ (\emph{i.e.}, $\varphi_i$ is satisfied by the assignment induced by $\pi$). The path $\pi$ is guaranteed to visit either a state labelled $x_i$ or a state labelled $\overline{x}_i$ but not both, by virtue of the subformula $(B)$. Therefore, the eventual satisfaction of $x_i$ is equivalent to globally satisfying $\neg x_i$, whereas the satisfaction of $\varphi_i$ can be asserted by 
  requiring that for any  clause some literal is globally absent from the labels on $\pi$.
  Thus the substitution performed on $(C)$ does not alter the set of states on which the formula is satisfied either. 
  Concluding, $K,x_k \models \psi_k \iff K,x_k \models \psi_k'$ for all odd $0 \leq k < 2n$.
  
  Now, if $k$ is even, then 
  \begin{align*}
    \neg \psi_k \equiv \A \bigg[ \underbrace{\F\Big( \bigvee_{i=1}^n \overline{x}_i \land
                        \A\big(\F\bigvee_{i=1}^n (s_i^{00} \lor s_i^{01} \lor s_i^{10} \lor s_i^{11}) \lor \G(\bigvee_{i=1}^n x_i \limplies \psi_{k-1})\big) \Big)}_{(A)} \hskip4em \\ 
                        \lor \ \underbrace{\vphantom{\bigvee_i}\F\Big( \bigvee_{i=1}^n c_i\Big)}_{(B)} \ 
                        \lor \ \underbrace{\bigvee_{i=1}^n \Big(\F x_i \land \bigvee_j \bigwedge_m \G \neg \ell_{i,j,m} \Big)}_{(C)}
                 \bigg].
  \end{align*}
  Here, $(A)$ asserts that on all paths $\pi$ there is a state $\pi_p$ such that $\overline{x}_i \in \eta(\pi_p)$ for some $1 \leq i \leq n$ and all paths $\pi'$ starting in $\pi_p$ eventually visit a state labelled with $s_i^{j}$, $1 \leq i \leq n$, $j \in \{00,01,10,11\}$, or satisfy $\psi_{k-1}$ whenever $x_i \in \eta(\pi_p)$ for some $1 \leq i \leq n$. By construction of $K$, this is equivalent to stating the all paths $\pi'$ either pass the state labelled $c_i$ and globally satisfy $c_i \lor \psi_{k-1}$ or do not pass the state labelled $c_i$. As for the states in $K$ the formula $\F\big(\bigvee_{i=1}^n\overline{x_i} \land \chi\big) \equiv \bigvee_{i=1}^n\F\big(\overline{x_i} \land \chi\big)$ is satisfied iff $\bigvee_{i=1}^n\G\big(\bigvee (\Phi\setminus\{\overline{x}_i\}) \lor \chi\big)$ is satisfied, 
  the set of states in $K$ on which the $\psi_k$ is satisfied remains unaltered when substituting $(A)$ with 
  $\bigvee_{1 \leq i \leq n}\G\big( \bigvee (\Phi\setminus\{\overline{x}_i\}) \lor \A\big( \G (\bigvee\Phi\setminus\{c_i\}) \lor \G(c_i \lor \psi_{k-1}\big) \big)$.
  
  Similarly, the set of states in $K$ on which the $\psi_k$ is satisfied remains unaltered when substituting $(B)$ with $\bigvee_{i=2}^n\G\bigvee(\Phi \setminus \{s_i^{00},s_i^{01},s_{i-1}^{01},s_{i-1}^{11}\})$, as any path in $K$ that visits a state labelled with some $c_i$ cannot pass via states labelled with $s_i^{00}$, $s_i^{01}$, $s_{i-1}^{01}$, or $s_{i-1}^{11}$. 
  
  Finally, the equivalence of $\psi_k$ with $\bigvee_{i=1}^n \big(\G\bigvee(\Phi\setminus\{\overline{x}_i\}) \land \bigvee_j \bigwedge_m \G \neg \ell_{i,j,m} \big)$ follows from arguments similar to those for the $(C)$ part in the case that $k$ is odd.
  We conclude that $K,x_k \models \psi_k \iff K,x_k \models \psi_k'$ for all $0 \leq k < 2n$.
  Hence, $\CTLplusposMC(\{\A,\E,\G\})$ is $\DeltaPtwo$-hard.
  
  For $T=\{\A,\E,\F\}$ similar modifications show that $\CTLplusposMC(T)$ is $\DeltaPtwo$-hard, too. This concludes to proof of Theorem~\ref{thm:ctl+pos-mc}.
\end{proof}

Finally consider $\ECTLplus$, the combination of $\ECTL$ and $\CTLplus$. 
One can adapt the above hardness and membership proofs to hold for $\Finfty$ and $\Ginfty$ instead of $\F$ and $\G$:
For example, to establish the $\DeltaPtwo$-hardness of $\ECTLplusposMC(T)$ in case $T=\{\A,\E,\Ginfty\}$ we modify $K$ such that the states labelled $x_n$ and $\overline{x}_n$ are reachable from $z_p$ and $\overline{z}_p$ and assert that (a) the path quantified by the outmost path quantifier in $\psi_k$, $1 \leq i < 2n$, additionally satisfies $\bigwedge_{i=1}^n (\Ginfty\neg x_i \lor \Ginfty \neg \overline{x}_i)$ and (b) whenever $\overline{x}_i$ is labelled, then there exists a path that all but a finite number of times satisfies $x_i$. The changes if $\Finfty$ is available instead of $\Ginfty$ follow by the duality principle of these operators.
For its model checking problem we hence obtain: 

\begin{cor}
  Let $T$ be a set of temporal operators containing at least one path quantifier and let 
  $T'$ by obtained from $T$ by substituting $\Finfty$ with $\F$ and $\Ginfty$ with $\G$.
  Then $\ECTLplusMC(T) \equivcd \CTLplusMC(T')$
  and $\ECTLplusposMC(T) \equivcd \CTLplusposMC(T)$.
\end{cor}

\section{Conclusion}\label{sect:conclusion}

We have shown (Theorem~\ref{thm2}) that model checking for $\CTLpos(T)$ 
is already \P-complete for most fragments of \CTL. 
Only for some weak fragments, model checking becomes easier: 
if $T \subseteq \{\EX,\EF\}$ or $T\subseteq \{\AX,\AG\}$, then $\CTLposMC(T)$ is \LOGCFL-complete.
In the case that no $\CTL$-operators are used, 
\NC1-completeness of evaluating propositional formulae applies. 
As a direct consequence (Theorem~\ref{thm1}), model checking for $\CTL(T)$ is \P-complete
for every nonempty $T$. 
This shows that for the majority of interesting fragments, model
checking $\CTL(T)$ is inherently sequential and cannot be 
sped up using parallelism.

While all the results above can be transferred to $\ECTL$ (Theorem~\ref{thm:ectl-mc}), 
$\CTLplus$ and $\ECTLplus$ exhibit different properties. 
For both logics, the general model checking problem was shown to be complete for $\DeltaPtwo$
in \cite{lamasc01}. 
Here we proved that model checking fragments of $\CTLplus(T)$ and
$\ECTLplus(T)$ for $T \subseteq \{\A,\E,\X\}$ remains tractable,
while the existential and the universal fragments of
$\CTLpluspos(T)$ and $\ECTLpluspos(T)$ containing temporal operators
other than $\X$ are complete for $\NP$ and $\co\NP$, respectively. 

  \begin{figure}
    \centering
    \includegraphics[width=0.5\linewidth]{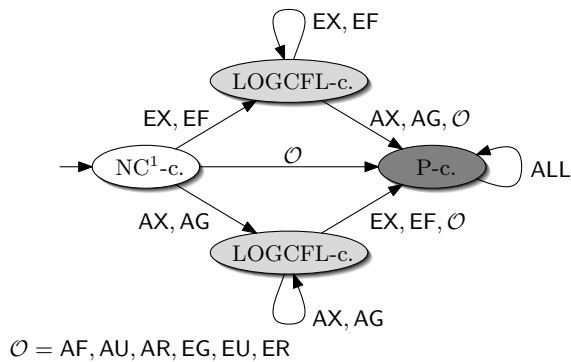}
    \caption{\label{fig:ctl-mc-BF}%
      Complexity of $\CTLposMC(T)$ 
      for all sets $T$ of \CTL-operators
      (depicted as a ``finite automaton''
       where states indicate completeness results and
       arrows indicate an increase of the set of \CTL-operators).
    }
  \end{figure}
  
Instead of restricting only the use of negation as done in this paper,
one might go one step further and restrict the allowed
Boolean connectives in an arbitrary way.  
One might, \emph{e.g.}, allow the exclusive-OR as the only propositional connective. 
This has been done for the case of linear temporal logic LTL in
\cite{bamuscscscvo07entcs}, where the complexity of $\LTLMC(T,B)$ for an
arbitrary set $T$ of temporal operators and $B$ of propositional
connectives was studied.
For example, restricting the Boolean connectives
to only one of the functions AND or OR leads to many
$\NL$-complete fragments in the presence of certain sets of temporal operators.
However a full classification is still open.

Considering the CTL variants considered in this paper, plus $\CTLstar$, over \emph{arbitrary} sets of Boolean operators
would be one way to generalise our results. In the case of $\CTLplus$ and $\CTLstar$,
where model checking is intractable \cite{emle87,schnoeb02,lamasc01},
such a more fine-grained complexity analysis could help draw a tighter border
between fragments with tractable and intractable model checking problems.
As for the corresponding satisfiability problems $\CTLSAT(T,B)$ and $\CTLstarSAT(T,B)$,
their complexity has been determined---with respect to the set of Boolean operators, completely---in \cite{memuthvo08}.

Throughout this paper, we have assumed that the formula \emph{and} the Kripke structure are part of the input and
can vary in size. The case where the complexity is measured in terms of the size of the formula (or the Kripke structure),
and the other component is assumed to be fixed, is usually referred to as \emph{specification complexity} (or \emph{system complexity}).
Our approach measures the joint complexity. In applications, where usually the structure is significantly bigger
than the specification, an analysis of the system complexity becomes interesting.
For system complexity, model checking for \CTL and \CTLstar is already \NL-complete \cite{BVW94,KuVaWo00}.
Still, the hope for a significant drop of system complexity justifies
a systematic analysis of fragments of these logics.

%


\bibliographystyle{alpha}
\bibliography{thi-hannover}

\end{document}